\documentclass[10pt, a4paper]{amsart}
\usepackage{lmodern}
\usepackage{amssymb}
\usepackage{amsthm}
\usepackage{cite}
\usepackage{txfonts}
\usepackage{enumitem}
\usepackage{dsfont}
\usepackage{graphicx}
\usepackage[usenames, dvipsnames]{color}
\usepackage{tikz}
\usetikzlibrary{automata, positioning, shapes}
\usetikzlibrary{shapes,snakes,calc}
\usetikzlibrary{decorations.pathreplacing}
\usepackage{mathtools}
\allowdisplaybreaks
\newtheorem{theorem}{Theorem}

\newtheorem{corollary}{Corollary}

\newtheorem{defn}{Definition}
\newtheorem{example}{Example}

\newtheorem{prob}{Problem}

\newtheorem{rem}{Remark}

\usepackage{multicol}

\newcommand{\norm}[1]{\left\lVert{#1}\right\rVert}
\newcommand{\abs}[1]{\left\lvert{#1}\right\rvert}

\newcommand{\pmat}[1]{\begin{pmatrix}#1\end{pmatrix}}

\newcommand{\R}{\mathbb{R}}
\newcommand{\N}{\mathbb{N}}
\renewcommand{\P}{\mathcal{P}}

\newcommand{\Dm}{\delta}
\newcommand{\Sd}{\mathcal{S}_{\delta}}
\newcommand{\Wd}{W_{\delta}}
\newcommand{\sWd}{\mathcal{W}_{\delta}}

\allowdisplaybreaks

\title[]{A new condition for stability of switched linear systems\\under restricted minimum dwell time switching}

\author{Atreyee Kundu}
\thanks{The author is with the Department of Electrical Engineering, Indian Institute of Science Bangalore, India. Email: atreyeek@iisc.ac.in.}
\thanks{Her research work is supported by the DST INSPIRE Faculty Award IFA17-ENG225 from the Department of Science and Technology, Govt. of India.}
\thanks{She thanks Debasish Chatterjee for helpful discussions.}
\date{\today}

\begin{document}

	\begin{abstract}
      		We propose matrix commutator based stability characterization for discrete-time switched linear systems under restricted switching. Given an admissible minimum dwell time, we identify sufficient conditions on subsystems such that a switched system is stable under all switching signals that obey the given restriction. The primary tool for our analysis is commutation relations between the subsystem matrices. Our stability conditions are robust with respect to small perturbations in the elements of these matrices. In case of arbitrary switching (i.e., given minimum dwell time \(= 1\)), we recover the prior result \cite[Proposition 1]{Agrachev'12} as a special case of our result.
	\end{abstract}
\maketitle
\begin{multicols}{2}
\section{Introduction}
\label{s:intro}
    A \emph{switched system} has two ingredients --- a family of systems and a switching signal. The switching signal selects an \emph{active subsystem} at every instant of time, i.e, the system from the family that is currently being followed \cite[\S 1.1.2]{Liberzon}. Switched systems find wide applications in power systems and power electronics, automotive control, aircraft and air traffic control, network and congestion control, etc. \cite[p.\ 5]{Sun}.
	
We consider a family of discrete-time linear systems
	\begin{align}
	\label{e:family}
		x(t+1) = A_{i}x(t),\:\:x(0)=x_{0},\:\:i\in\P,\:\:t\in\N_{0},
	\end{align}
	where \(x(t)\in\R^{d}\) is the vector of states at time \(t\), \(\P = \{1,2,\ldots,N\}\) is an index set, and \(A_{i}\in\R^{d\times d}\), \(i\in\P\) are constant matrices. Let \(\sigma:\N_{0}\to\P\) be a {switching signal}. A discrete-time {switched linear system} generated by the family of systems \eqref{e:family} and a switching signal \(\sigma\) is described as
	\begin{align}
	\label{e:swsys}
		x(t+1) = A_{\sigma(t)}x(t),\:\:x(0) = x_{0},\:\:t\in\N_{0}.
	\end{align}
	The solution to \eqref{e:swsys} is given by
	\[
		x(t) = A_{\sigma(t-1)}\cdots A_{\sigma(1)}A_{\sigma(0)}x_{0},\:\:t\in\N,
	\]
	where we have suppressed the dependence of \(x\) on \(\sigma\) for notational convenience.

    In this paper we will work with switching signals that obey a pre-specified minimum dwell time \(\Dm\in\N\) on every subsystem \(i\in\P\), i.e., whenever a subsystem \(i\) is activated by \(\sigma\), it remains active for at least \(\Dm\) units of time. In many engineering applications, a restriction on minimum dwell time on subsystems is natural. For instance, actuator saturations may prevent switching frequency beyond a certain limit, or in order to switch from one component to another, a system may undergo certain operations of non-negligible durations leading to a minimum dwell time requirement on each subsystem \cite{Heydari'17}. Let \(0=:\tau_{0}<\tau_{1}<\cdots\) be the \emph{switching instants}; these are the points in time when \(\sigma\) switches from one subsystem to another. Our switching signals satisfy: there exists \(\Dm\in\N\) such that the following condition holds:
	\begin{align}
	\label{e:sw_res}
		\tau_{i+1}-\tau_{i}\geq\Dm,\:\:i=0,1,2,\ldots.
	\end{align}
	Let \(\Sd\) denote the set of all switching signals \(\sigma\) that satisfy condition \eqref{e:sw_res}. Our focus is on global uniform exponential stability (GUES) of the switched system \eqref{e:swsys}.
	\begin{defn}{\cite[Section 1]{Agrachev'12}}
	\label{d:gues}
    \rm{
		The switched system \eqref{e:swsys} is \emph{globally uniformly exponentially stable (GUES) over the set of switching signals \(\Sd\)} if there exist positive numbers \(c\) and \(\lambda\) such that for arbitrary choices of the initial condition \(x_{0}\) and the switching signal \(\sigma\in\Sd\), the following condition holds:
		\begin{align}
		\label{e:gues1}
			\norm{x(t)} \leq c e^{-\lambda t}\norm{x_{0}}\:\:\text{for all}\:t\in\N,
		\end{align}
		where \(\norm{v}\) denotes the Euclidean norm of a vector \(v\).
    }
	\end{defn}
    The term `uniform' in the above definition indicates that the numbers \(c\) and \(\lambda\) can be selected independent of \(\sigma\). Fix a \(\sigma\in\Sd\). Let \(\Wd\) denote the corresponding matrix product defined as: \(\Wd = \cdots A_{\sigma(2)}A_{\sigma(1)}A_{\sigma(0)}\), and \(\sWd\) be the set of all products corresponding to the switching signals \(\sigma\in\Sd\). We let \(\overline{W}\) denote an initial segment of \(\Wd\), and we will use \(\abs{\overline{W}}\) to denote the length of \(\overline{W}\), i.e., the number of matrices that appear in \(\overline{W}\), counting repetitions. Condition \eqref{e:gues1} can be written equivalently as \cite[Section 2]{Agrachev'12}: for every \(\Wd\in\sWd\) and every initial segment \(\overline{W}\) of \(\Wd\), the following condition holds:
	\begin{align}
	\label{e:gues2}
		\norm{\overline{W}}\leq c e^{-\lambda\abs{\overline{W}}}.
	\end{align}
	We will solve the following problem:
	\begin{prob}
	\label{prob:mainprob}
    \rm{
		Given a minimum dwell time \(\Dm\in\N\), find conditions on the matrices \(A_{i}\), \(i\in\P\), such that the switched system \eqref{e:swsys} is GUES over the set of switching signals \(\Sd\).
    }
	\end{prob}

    We will rely on matrix commutators (Lie brackets) of the subsystem matrices to solve Problem \ref{prob:mainprob}. It is well-known that a switched linear system is stable under arbitrary switching (i.e., when \(\delta = 1\)) if all subsystems are Schur stable and commute pairwise \cite{Narendra'94} or are ``sufficiently close'' to a set of matrices whose elements commute pairwise \cite{Agrachev'12}. On the one hand, these conditions are only sufficient, and their non-satisfaction does not guarantee that a switched system is not stable under all switching signals. On the other hand, a switched system that is not stable under arbitrary switching, may be stable under sets of switching signals that obey a certain minimum dwell time. The above features motivate us to derive matrix commutator conditions for stability of \eqref{e:swsys} under the elements of \(\Sd\). Towards this end, we follow the combinatorial analysis technique presented in \cite{Agrachev'12}. In particular, we split matrix products \(\overline{W}\) into sums and apply counting arguments on them.

    Our stability conditions involve upper bounds on the norms of the commutators of the subsystem matrices and a set of scalars relating to the individual matrices and the given minimum dwell time. These conditions also possess inherent robustness in the sense that if the elements of the subsystem matrices are perturbed by a small margin such that the matrices are not ``too far'' from a set of matrices for which certain products commute, then stability of the switched system \eqref{e:swsys} remains preserved under switching signals obeying a minimum dwell time. For \(\Dm = 1\) (i.e., the case of arbitrary switching), we recover \cite[Proposition 1]{Agrachev'12} as a special case of our result.

    The remainder of this paper is organized as follows: in \S\ref{s:prelims} we catalog the tools required for our analysis. Our results appear in \S\ref{s:mainres}. We also describe various features of our results in this section. We conclude in \S\ref{s:concln}.
	
\section{Preliminaries}
\label{s:prelims}
    Since the set \(\Sd\) includes constant switching signals, a necessary condition for GUES over \(\Sd\) is that all subsystem matrices \(A_{i}\), \(i\in\P\), are Schur stable. This implies that there exists \(\N\ni m\geq \Dm\) such that the following condition holds:
    \begin{align}
    \label{e:key_ineq1}
        \norm{A_{i}^{m}}\leq\rho < 1,\:\:\text{for all}\:i\in\P.
    \end{align}
    Of course, the choice of such \(m\) is not unique; we use the smallest \(m\geq\Dm\) that satisfies \eqref{e:key_ineq1}. Schur stability of the matrices \(A_{i}\), \(i\in\P\) is, however, not sufficient to guarantee stability of \eqref{e:swsys} under all elements of \(\Sd\), \(\Dm\in\N\) given, see e.g., \cite[\S 3.2.1]{Liberzon}.

    Let
    \begin{align}
    \label{e:M_defn}
        M = \max_{i\in\P}\norm{A_{i}},
    \end{align}
    \begin{align}
    \label{e:K1_defn}
        K_{1} = \lfloor\frac{\Dm}{m}\rfloor,
    \end{align}
    \begin{align}
    \label{e:K2_defn}
        K_{2} = \big\lfloor\frac{(N-1)(m-1)}{\Dm}\big\rfloor,
    \end{align}
    and
    \begin{align}
    \label{e:K3_defn}
        K_{3} = (N-1)(m-1) - K_{2}\Dm,
    \end{align}
    where for \(y\in\R\), \(\lfloor y\rfloor\) denotes the greatest integer less than or equal to \(y\). We define the commutators of the matrix products \(A_{i}^{p}\) and \(A_{j}^{q}\), \(p,q\in\{1,\Dm\}\) as follows:
    \begin{align}
    \label{e:commutator_defn}
        E_{ij}^{p,q} = A_{i}^{p}A_{j}^{q} - A_{j}^{q}A_{i}^{p},\:\:i,j\in\P.
    \end{align}
        The use of commutators of the matrix products \(A_{i}^{p}\) and \(A_{j}^{q}\), \(p,q\in\{1,\Dm\}\), \(i,j\in\P\) instead of commutators of the matrices \(A_{i}\) and \(A_{j}\), \(i,j\in\P\) as employed in \cite{Agrachev'12}, is motivated by the structure of our switching signals \(\sigma\in\Sd\), see Remark \ref{rem:analysis} for a detailed discussion.
    We are now in a position to present our results.
\section{Results}
\label{s:mainres}
    The following theorem identifies sufficient conditions on the subsystem matrices \(A_{i}\), \(i\in\P\), such that the switched system \eqref{e:swsys} is GUES over \(\Sd\).
    \begin{theorem}
    \label{t:mainres1}
        Consider the family of systems \eqref{e:family}. Let \(\Dm\in\N\) be given, the matrices \(A_{i}\), \(i\in\P\), satisfy \eqref{e:key_ineq1} with \(m\geq\Dm\), and \(\lambda\) be an arbitrary positive number satisfying
        \begin{align}
        \label{e:maincondn1}
            \rho e^{\lambda m} < 1.
        \end{align}
        Suppose that there exist \(\varepsilon_{p,q}\), \(p,q\in\{1,\Dm\}\) small enough such that the following conditions hold:
        \begin{align}
        \label{e:maincondn2}
            \norm{E_{ij}^{p,q}}\leq\varepsilon_{p,q}\:\:\text{for all}\:i,j\in\P,
        \end{align}
        and
        \begin{align}
        \label{e:maincondn3}
            &\rho e^{\lambda m} + \Biggl(K_{1}K_{2}\varepsilon_{\Dm,\Dm}M^{(N-1)(m-1)+m-2\Dm}\nonumber\\
            &+ K_{1}K_{3}\varepsilon_{\Dm,1}M^{(N-1)(m-1)+m-\Dm-1}\nonumber\\
            &+ (m-K_{1}\Dm)K_{2}\varepsilon_{1,\Dm}M^{(N-1)(m-1)+m-\Dm-1}\nonumber\\
            &+ (m-K_{1}\Dm)K_{3}\varepsilon_{1,1}M^{(N-1)(m-1)+m-2}\Biggr)
            \times e^{\lambda\bigl(N(m-1)+1\bigr)}\leq 1,
        \end{align}
        where \(M\), \(K_{1}\), \(K_{2}\), \(K_{3}\) and \(E_{ij}^{p,q}\), \(p,q\in\{1,\Dm\}\), \(i,j\in\P\) are as defined in \eqref{e:M_defn}, \eqref{e:K1_defn}, \eqref{e:K2_defn}, \eqref{e:K3_defn} and \eqref{e:commutator_defn}, respectively. Then the switched system \eqref{e:swsys} is GUES over the set of switching signals \(\Sd\).
    \end{theorem}

    \begin{proof}
        It suffices to show that if the conditions of Theorem \ref{t:mainres1} hold, then there exists a positive number \(c\) such that \eqref{e:gues2} holds for every initial segment \(\overline{W}\) of every \(\Wd\in\sWd\). We will employ mathematical induction on the length of an initial segment \(\overline{W}\) of \({\Wd}\) to establish \eqref{e:gues2}.

        \emph{A. Induction basis}: Pick \(c\) large enough so that \eqref{e:gues2} holds with all \(\overline{W}\) satisfying \(\abs{\overline{W}}\leq N(m-1)+1\).

        \emph{B. Induction hypothesis}: Let \(\abs{\overline{W}}\geq N(m-1)+2\) and assume that \eqref{e:gues2} is proved for all products of length less than \(\abs{\overline{W}}\).

        \emph{C. Induction step}: Let \(\overline{W} = LR\), where \(\abs{L} = N(m-1)+1 = (N-1)(m-1)+m\).
        We claim that there exists an index \(i\in\P\) such that \(L\) contains at least \(m\)-many \(A_{i}\)'s.
        Indeed, it follows from the fact that \((N-1)(m-1)+m \geq m\) and there are \(N\) subsystems. Without loss of generality, let \(i=1\). We rewrite \(L\) as
        \[
            L = A_{1}^{m}L_{1} + L_{2},
        \]
        where \(\abs{L_{1}} = (N-1)(m-1)\). The term \(L_{2}\) contains at most
        \begin{itemize}[label = \(\circ\), leftmargin = *]
            \item \(K_{1}K_{2}\) terms of length \((N-1)(m-1)+m-2\Dm+1\) with \((N-1)(m-1)+m-2\Dm\) \(A_{i}\)'s and \(1\) \(E_{1i}^{\Dm,\Dm}\),
            \item \(K_{1}K_{3}\) terms of length \((N-1)(m-1)+m-\Dm\) with \((N-1)(m-1)+m-\Dm-1\) \(A_{i}\)'s and \(1\) \(E_{i1}^{\Dm,1}\),
            \item \((m-K_{1}\Dm)K_{2}\) terms of length \((N-1)(m-1)+m-\Dm\) with \((N-1)(m-1)+m-\Dm-1\) \(A_{i}\)'s and \(1\) \(E_{1i}^{1,\Dm}\), and
            \item \((m-K_{1}\Dm)K_{3}\) terms of length \((N-1)(m-1)+m-1\) with \((N-1)(m-1)+m-2\) \(A_{i}\)'s and \(1\) \(E_{1i}^{1,1}\).
        \end{itemize}

        Now, from the sub-multiplicativity and sub-additivity properties of the induced norm, we have
        \begin{align}
        \label{e:pf1_step1}
            \norm{\overline{W}} &= \norm{LR} \leq \norm{A_{1}^{m}}\norm{L_{1}R}+\norm{L_{2}}\norm{R}\nonumber\\
            &\leq \rho c e^{-\lambda(\abs{\overline{W}}-m)}+ \Biggl(K_{1}K_{2}\varepsilon_{\Dm,\Dm}M^{(N-1)(m-1)+m-2\Dm}\nonumber\\
            &+ K_{1}K_{3}\varepsilon_{\Dm,1}M^{(N-1)(m-1)+m-\Dm-1}\nonumber\\
            &+ (m-K_{1}\Dm)K_{2}\varepsilon_{1,\Dm}M^{(N-1)(m-1)+m-\Dm-1}\nonumber\\
            &+ (m-K_{1}\Dm)K_{3}\varepsilon_{1,1}M^{(N-1)(m-1)+m-2}\Biggr)\nonumber\\
            &\hspace*{0.2cm}\times c e^{-\lambda\bigl(\abs{\overline{W}}-(N(m-1)+1)\bigr)}\nonumber\\
            &=ce^{-\lambda\abs{\overline{W}}}\Biggl(\rho e^{\lambda m} + \Big( K_{1}K_{2}\varepsilon_{\Dm,\Dm}M^{(N-1)(m-1)+m-2\Dm}\nonumber\\
            &+ K_{1}K_{3}\varepsilon_{\Dm,1}M^{(N-1)(m-1)+m-\Dm-1}\nonumber\\
            &+ (m-K_{1}\Dm)K_{2}\varepsilon_{1,\Dm}M^{(N-1)(m-1)+m-\Dm-1}\nonumber\\
            &+ (m-K_{1}\Dm)K_{3}\varepsilon_{1,1}M^{(N-1)(m-1)+m-2}\Big)
            \times e^{\lambda\bigl(N(m-1)+1\bigr)}\Biggr).
        \end{align}
        The upper bounds on \(\norm{L_{1}R}\) and \(\norm{R}\) are obtained from the relations \(\abs{\overline{W}} = \abs{A_{1}^{m}}+\abs{L_{1}R}\) and \(\abs{\overline{W}} = \abs{L}+\abs{R}\), respectively. Applying \eqref{e:maincondn3} to \eqref{e:pf1_step1} leads to \eqref{e:gues2}. Consequently, \eqref{e:swsys} is GUES over \(\Sd\).
%
    \end{proof}

    Theorem \ref{t:mainres1} characterizes a subset of the set of all Schur stable matrices that preserves stability of a switched system under all switching signals that obey a pre-specified minimum dwell time on all subsystems. The characterization involves upper bounds on the matrix norms of the commutators of the matrix products \(A_{i}^{p}\) and \(A_{j}^{q}\), \(p,q\in\{1,\Dm\}\), \(i,j\in\P\), and the scalars \(M\), \(m\), \(K_{1}\), \(K_{2}\), \(K_{3}\), which are related to the matrices \(A_{i}\), \(i\in\P\), the total number of subsystems \(N\), and the given minimum dwell time \(\Dm\). Notice that given \(m\geq\Dm\) and \(\rho < 1\), there always exists a positive scalar \(\lambda\) such that \eqref{e:maincondn1} holds. We further rely on the existence of small enough \(\varepsilon_{p,q}\), \(p,q\in\{1,\Dm\}\) that satisfy \eqref{e:maincondn2}-\eqref{e:maincondn3}. The scalars \(\varepsilon_{p,q}\), \(p,q\in\{1,\Dm\}\) give a measure of the ``closeness'' of the set of matrices \(A_{i}\), \(i\in\P\) to a set of matrices for which the matrix products under consideration commute. They associate an inherent ``robustness'' to our stability conditions in the sense that if the elements of \(A_{i}\), \(i\in\P\) are perturbed (e.g., if we are relying on approximate models of the subsystems, or the parameters of the subsystems are prone to evolve over time) in a manner so that conditions \eqref{e:key_ineq1}, \eqref{e:maincondn2}-\eqref{e:maincondn3} continue to hold, then GUES of \eqref{e:swsys} remains preserved under the elements of \(\Sd\). If the subsystem matrices \(A_{i}\) and \(A_{j}\) strictly commute for all \(i,j\in\P\), then we have that the matrices \(A_{i}^{p}\) and \(A_{j}^{q}\) commute for all \(p,q\in\{1,\delta\}\) and all \(i,j\in\P\) \cite[Fact 2.18.3]{Bernstein}. Consequently, \(\varepsilon_{pq} = 0\) for all \(p,q\in\{1,\delta\}\), and condition \eqref{e:maincondn3} holds in view of \eqref{e:maincondn1}. 

    \begin{rem}
    \label{rem:analysis}
    \rm{
    Our analysis technique differs from \cite[Proof of Proposition 1]{Agrachev'12} in the following way: in \cite{Agrachev'12} to arrive at \(A_{1}^{m}L_{1}+L_{2}\), the authors split a matrix product into sums by exchanging at every step two matrices \(A_{1}\) and \(A_{i}\), \(i\neq 1\), that appear consecutively in \(L\). This procedure leads to the usage of the commutators \(E_{ij}^{1,1}\) and a maximum of \(m(N-1)(m-1)\) terms in \(L_{2}\), each of length \(N(m-1)\) containing \(1\) \(E_{ij}^{1,1}\), \(i\neq 1\) and \(N(m-1)-1\) \(A_{i}\)'s. Consider, for example, \(N=3\) and \(m = 3\). Let \(L = A_{3}^{2}A_{2}^{2}A_{1}^{3}\). It can be rewritten as
    \begin{align*}
        L &= A_{1}^{3}A_{3}^{2}A_{2}^{2} - A_{1}A_{1}E_{13}^{1,1}A_{3}A_{2}A_{2} - A_{1}A_{1}A_{3}E_{13}^{1,1}A_{2}A_{2}\\
        &\:\:- A_{1}A_{1}A_{3}A_{3}E_{12}^{1,1}A_{2}
         - A_{1}A_{1}A_{3}A_{3}A_{2}E_{12}^{1,1} - A_{1}E_{13}^{1,1}A_{3}A_{2}A_{2}A_{1}\\
        &\:\:- A_{1}A_{3}E_{13}^{1,1}A_{2}A_{2}A_{1}
         - A_{1}A_{3}A_{3}E_{12}^{1,1}A_{2}A_{1} -
        A_{1}A_{3}A_{3}A_{2}E_{12}^{1,1}A_{1}\\
        &\:\:- E_{13}^{1,1}A_{3}A_{2}A_{2}A_{1}A_{1}
         - A_{3}E_{13}^{1,1}A_{2}A_{2}A_{1}A_{1} -
        A_{3}A_{3}E_{12}^{1,1}A_{2}A_{1}A_{1}\\
        &\:\: - A_{3}A_{3}A_{2}E_{12}^{1,1}A_{1}A_{1}.
    \end{align*}
    In contrast, in this paper we utilize the minimum dwell time property of switching signals to arrive at the desired structure of \(L\). Our procedure involves exchanging \(K_{1}\) products of length \(\Dm\) of \(A_{1}\) with at most \(K_{2}\) products of length \(\Dm\) and \(K_{3}\) entries of matrix \(A_{i}\), \(i\neq 1\), and \(m-K_{1}\Dm\) entries of \(A_{1}\) with at most \(K_{2}\) products of length \(\Dm\) and \(K_{3}\) entries of matrix \(A_{i}\), \(i\neq 1\). This leads us to the usage of the commutators \(E_{ij}^{\Dm,\Dm}\), \(E_{ij}^{1,\Dm}\), \(E_{ij}^{\Dm,1}\) and \(E_{ij}^{1,1}\) unlike only \(E_{ij}^{1,1}\) employed in \cite{Agrachev'12}. Consider the above example with \(\Dm = 2\). We rewrite \(L\) as
        \begin{align*}
            L &= \underline{A_{3}^2A_{1}}A_{2}^{2}A_{1}^{2} - A_{3}^2 E_{12}^{1,2}A_{1}^2\\
            &= A_{1}A_{3}^2\underline{A_{2}^2A_{1}^2} - E_{13}^{1,2}A_{2}^2A_{1}^2 - A_{3}^2E_{12}^{1,1}A_{1}^2\\
            &= A_{1}\underline{A_{3}^2A_{1}^2}A_{2}^2 - A_{1}A_{3}^2E_{12}^{2,2} - E_{13}^{1,2}A_{2}^2A_{1}^2-A_{3}^2E_{12}^{1,2}A_{1}^2\\
            &= A_{1}^3A_{3}^2A_{2}^2 - A_{1}E_{13}^{2,2}A_{2}^2 - A_{1}A_{3}^2E_{12}^{2,2} - E_{13}^{1,2}A_{2}^2A_{1}^2 - A_{3}^2E_{12}^{1,2}A_{1}^2.
        \end{align*}
    See Example \ref{ex:main_ex} for a setting where the conditions of \cite[Proposition 1]{Agrachev'12} do not hold, but the conditions of Theorem \ref{t:mainres1} are satisfied with \(\Dm = 2\). Consequently, while we cannot conclude stability of \eqref{e:swsys} under arbitrary switching based on \cite[Proposition 1]{Agrachev'12}, we arrive at stability of \eqref{e:swsys} under all switching signals obeying a minimum dwell time of \(2\) units on all subsystems. Hence, Theorem \ref{t:mainres1} is useful to switched systems for which stability under arbitrary switching signals cannot be verified and/or stability under switching signals obeying a certain minimum dwell time is of relevance.
    }
    \end{rem}

    It is important to note that Theorem \ref{t:mainres1} provides only sufficient conditions for GUES of \eqref{e:swsys} in the sense that non-satisfaction of conditions \eqref{e:maincondn2}-\eqref{e:maincondn3} does not imply instability of \eqref{e:swsys} under \(\sigma\in\Sd\). Indeed, we consider properties of the subsystem matrices and their commutators that lead to \eqref{e:gues2}, but do not look for conditions on them that are implied by \eqref{e:gues2}. For \(\Dm = 1\) (i.e., the case of arbitrary switching), we recover \cite[Proposition 1]{Agrachev'12} as a special case of Theorem \ref{t:mainres1}.
    \begin{corollary}
    \label{cor:arbitrary_sw}
        Consider the family of systems \eqref{e:family}. Let \(\Dm = 1\), the matrices \(A_{i}\), \(i\in\P\), satisfy \eqref{e:key_ineq1} with \(m\geq\Dm\), and \(\lambda\) be an arbitrary positive number satisfying \eqref{e:maincondn1}. Suppose that there exists \(\varepsilon\) small enough such that
        \begin{align}
        \label{e:maincondn4}
            \norm{E_{ij}^{1,1}}\leq\varepsilon\:\:\text{for all}\:i,j\in\P,
        \end{align}
        and
        \begin{align}
        \label{e:maincondn5}
            \rho e^{\lambda m} + m(N-1)(m-1)\varepsilon M^{N(m-1)+1}\times e^{\lambda\bigl(N(m-1)+1\bigr)}\leq 1.
        \end{align}
        Then the switched system \eqref{e:swsys} is GUES over the set of switching signals \(\Sd\).
    \end{corollary}

    \begin{proof}
        From Theorem \ref{t:mainres1}, we have that the switched system \eqref{e:swsys} is GUES if conditions \eqref{e:maincondn2}-\eqref{e:maincondn3} hold.

        Given \(\Dm = 1\), we have
        \(
            K_{1} = m,\:\: m-K_{1}\Dm = 0,\:\: K_{2} = (N-1)(m-1),\:\:\text{and}\:K_{3} = 0.
        \)
        Consequently,
        \begin{align*}
            & K_{1}K_{2}\varepsilon_{\Dm,\Dm}M^{(N-1)(m-1)+m-2\Dm}\\
            &= m(N-1)(m-1)\varepsilon_{1,1}M^{(N-1)(m-1)+m-2}\\
            &= m(N-1)(m-1)\varepsilon_{1,1}M^{N(m-1)-1},\\
            &K_{1}K_{3}\varepsilon_{\Dm,1}M^{(N-1)(m-1)+m-\Dm-1} = 0,\\
            &(m-K_{1}\Dm)K_{2}\varepsilon_{1,\Dm}M^{(N-1)(m-1)+m-\Dm-1} = 0,\\
            &(m-K_{1}\Dm)K_{3}\varepsilon_{1,1}M^{(N-1)(m-1)+m-2} = 0.
        \end{align*}

        Now, condition \eqref{e:maincondn2} becomes
        \[
            \norm{E_{ij}^{1,1}}\leq\varepsilon_{1,1}\:\:\text{for all}\:i,j\in\P,
        \]
        and condition \eqref{e:maincondn3} becomes
        \begin{align*}
            \rho e^{\lambda m} &+ m(N-1)(m-1)\varepsilon_{1,1}M^{N(m-1)-1}\times e^{\lambda\bigl(N(m-1)+1\bigr)} \leq 1.
        \end{align*}
        Setting \(\varepsilon_{1,1} = \varepsilon\) completes our proof of Corollary \ref{cor:arbitrary_sw}.
    \end{proof}
    \begin{rem}
    \label{rem:min-max_algo}
    \rm{
        A vast body of the switched systems literature is devoted to finding estimates of stabilizing dwell times, see e.g., \cite{ghi,Zhai'02,xyz} and the references therein. In contrast, in the current paper, we deal with guaranteeing stability under a ``given'' minimum dwell time. We refer the reader to \cite{Jungers'13, Heydari'17, abc, def} for existing results on stability and optimal control of switched systems under pre-specified restrictions on minimum dwell time. The main difference of our results with respect to the prior works is in the use of commutation relations between the subsystem matrices in contrast to the widely used multiple Lyapunov-like functions \cite{Branicky'98} based stability analysis under restricted switching. On the one hand, we avoid the design of Lyapunov-like functions corresponding to subsystems that satisfy certain conditions individually and among themselves (see. e.g., the techniques of \cite{abc,def}) and rely directly on the properties of the subsystem matrices to guarantee stability. On the other hand, our stability conditions are limited to the case of switched linear systems unlike Lyapunov-like functions based techniques that extend to switched nonlinear systems under standard assumptions.
    }
    \end{rem}

    \begin{example}
    \label{ex:main_ex}
    \rm{
    Consider \(\P = \{1,2\}\) with \(A_{1} = \pmat{0.02 & 0.93\\-0.53 & -0.92}\) and \(A_{2} = \pmat{0.04 & 0.09\\0.08 & -0.11}\). Clearly, both the subsystems are Schur stable.

    We have
       \( \norm{A_{1}^2} = 1.1204, \norm{A_{1}^3} = 0.5404,
        \norm{A_{2}^2} = 0.0220\),\\ \(\norm{A_{2}^3} = 0.0033.\)
    Therefore, the smallest integer \(m\) for which \eqref{e:key_ineq1} holds is \(m=3\). Also, \(\rho = 0.5404\). Let \(\lambda = 0.01\) leading to
    \(
        \rho e^{\lambda m} = 0.5569 < 1.
    \)

    Now, \(M = \max\{\norm{A_{1}},\norm{A_{2}}\} = 1.3683\), \(K_{1} = \lfloor\frac{m}{\Dm}\rfloor = 1\), \(K_{2} = \lfloor\frac{(N-1)(m-1)}{\Dm}\rfloor = 1\), \(K_{3} = (N-1)(m-1) - K_{2}\Dm = 0\), \(\varepsilon_{\Dm,\Dm} = 0.0133\), \(\varepsilon_{\Dm,1} = 0.1897\), \(\varepsilon_{1,\Dm} = 0.1897\), \(\varepsilon_{1,1} = 0.2108\).

    The conditions of \cite[Proposition 1]{Agrachev'12} do not hold in the setting described above. Indeed,
        \(\rho e^{\lambda m} + m(N-1)(m-1)\varepsilon_{1,1}M^{N(m-1)+1}\times e^{\lambda\bigl(N(m-1)+1\bigr)}
        =6.9513 > 1\).
    Consequently, we cannot conclude about (in)stability of \eqref{e:swsys} under arbitrary switching solely based on commutation relations between the subsystem matrices.

    Now, let \(\Dm = 2\). We have that the left-hand side of inequality \eqref{e:maincondn3} computes to \(0.9664\)
    implying that the switched system \eqref{e:swsys} is GUES under all switching signals that obey a minimum dwell time \(\Dm = 2\).

    We next perturb the elements of the matrices \(A_{1}\) and \(A_{2}\) to generate
	\begin{align*}
		\tilde{A}_{1} &= A_{1} + \pmat{0.03 & 0.02\\-0.07 & 0} = \pmat{0.05 & 0.95\\-0.6 & -0.92},\:\:\text{and}\\
		\tilde{A}_{2}&= A_{2} + \pmat{0 & 0\\0.02 & 0} = \pmat{0.04 & 0.09\\0.1 & -0.11}.
	\end{align*}
	
	We have
		\(\norm{A_{1}^2} = 1.1384,\norm{A_{1}^3} = 0.5180,
		\norm{A_{2}^2} = 0.0243\), \(\norm{A_{2}^3} = 0.0038\)
	leading to \(\rho = 0.5180\) and \(m = 3\). Choosing \(\lambda = 0.0001\) gives
	\(
		\rho e^{\lambda m} = 0.5182 < 1.
	\)
	Also, \(M = \max\{\norm{A_{1}},\norm{A_{2}}\} = 1.4043\), \(K_{1} = K_{2} = 1\), \(K_{3} = 0\), \(\varepsilon_{\Dm,\Dm} = 0.0157\), \(\varepsilon_{\Dm,1} = 0.2244\), \(\varepsilon_{1,\Dm} = 0.2244\), \(\varepsilon_{1,1} = 0.2579\). Consequently, the numerical value of the expression on the left-hand side of inequality \eqref{e:maincondn3} is \(0.9830\).

    Both for the subsystems \(\{A_{1},A_{2}\}\) and \(\{\tilde{A}_{1},\tilde{A}_{2}\}\), we generate \(1000\) random switching signals that obey a minimum dwell time \(\Dm = 2\) on both the subsystems \(1\) and \(2\), and plot the corresponding \((\norm{x(t)})_{t\in\N_{0}}\) in Figures \ref{fig:ex_xplot} and \ref{fig:ex_xplot_rob}, respectively. The initial conditions \(x_{0}\) are chosen from the interval \([-100,100]^{2}\) uniformly at random.
    }
    \end{example}

	\section{Conclusion}
\label{s:concln}
	In this paper we presented sufficient conditions on subsystems such that stability of a switched linear system is preserved under every switching signal that obeys a ``given'' minimum dwell time. Our characterization of stability involves commutation relations between the subsystem matrices. Since we dealt with stability of a switched system under all switching signals obeying a given minimum dwell time, the overarching assumption has been Schur stability of all the subsystem matrices. However, in practice unstable subsystems are often encountered. In addition, the maximum dwell time on all subsystems may also be restricted, see e.g., \cite{abc, def} for examples. In \cite{min-max19} we have reported a matrix commutator based characterization of switching signals that activate both stable and unstable subsystems, obey pre-specified restrictions on both minimum and maximum dwell times, and preserve stability of the resulting switched system.
	


\end{multicols}

\begin{figure}[!htb]
    \centering
    \begin{minipage}{.5\textwidth}
        \centering
        \includegraphics[scale = 0.3]{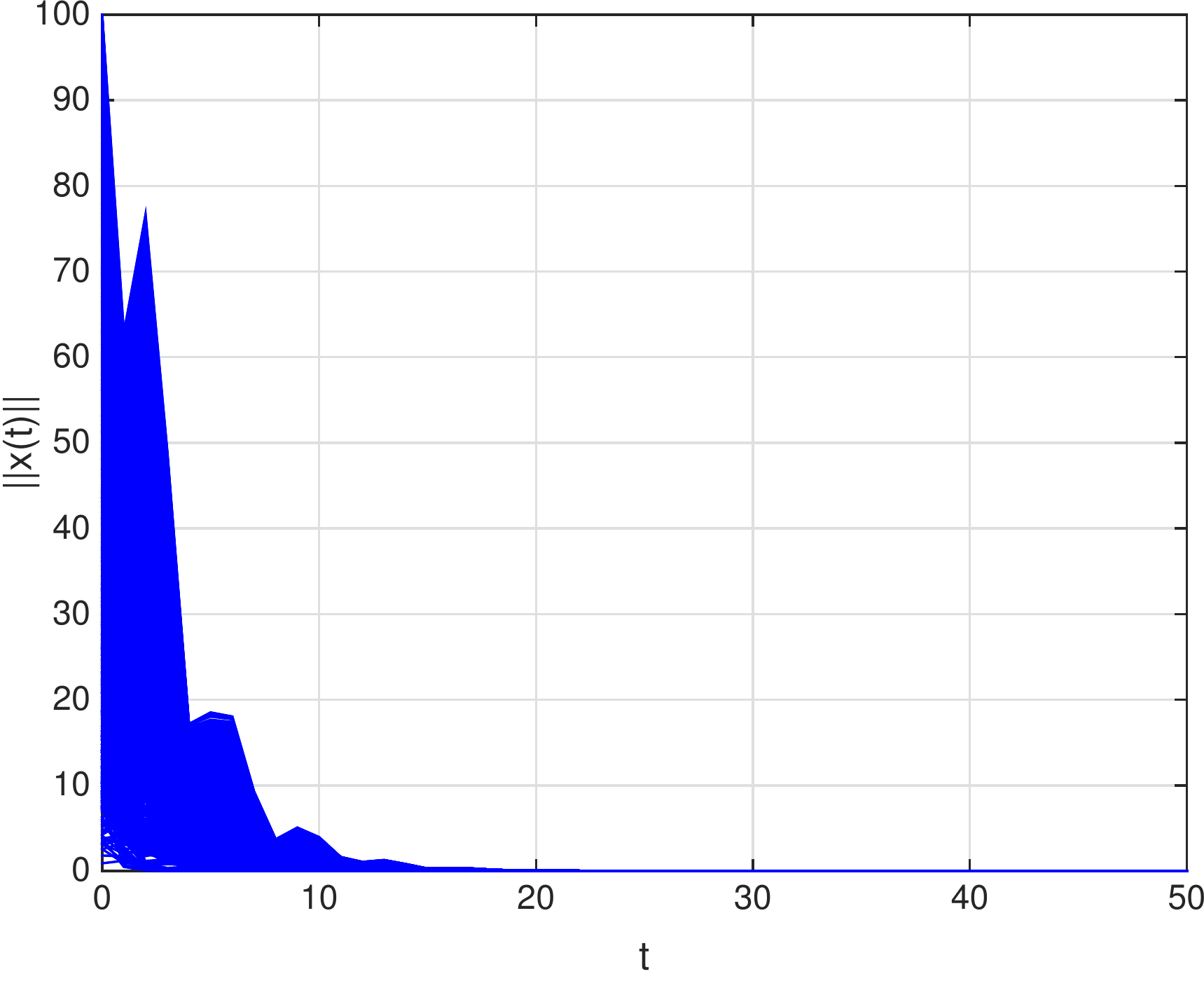}
        \caption{Plot of \((\norm{x(t)})_{t\in\N_{0}}\) with subsystems \(A_{1}\) and \(A_{2}\) under \(\sigma\in\Sd\), \(\Dm = 2\)}\label{fig:ex_xplot}
    \end{minipage}%
    \begin{minipage}{0.5\textwidth}
        \centering
        \includegraphics[scale = 0.3]{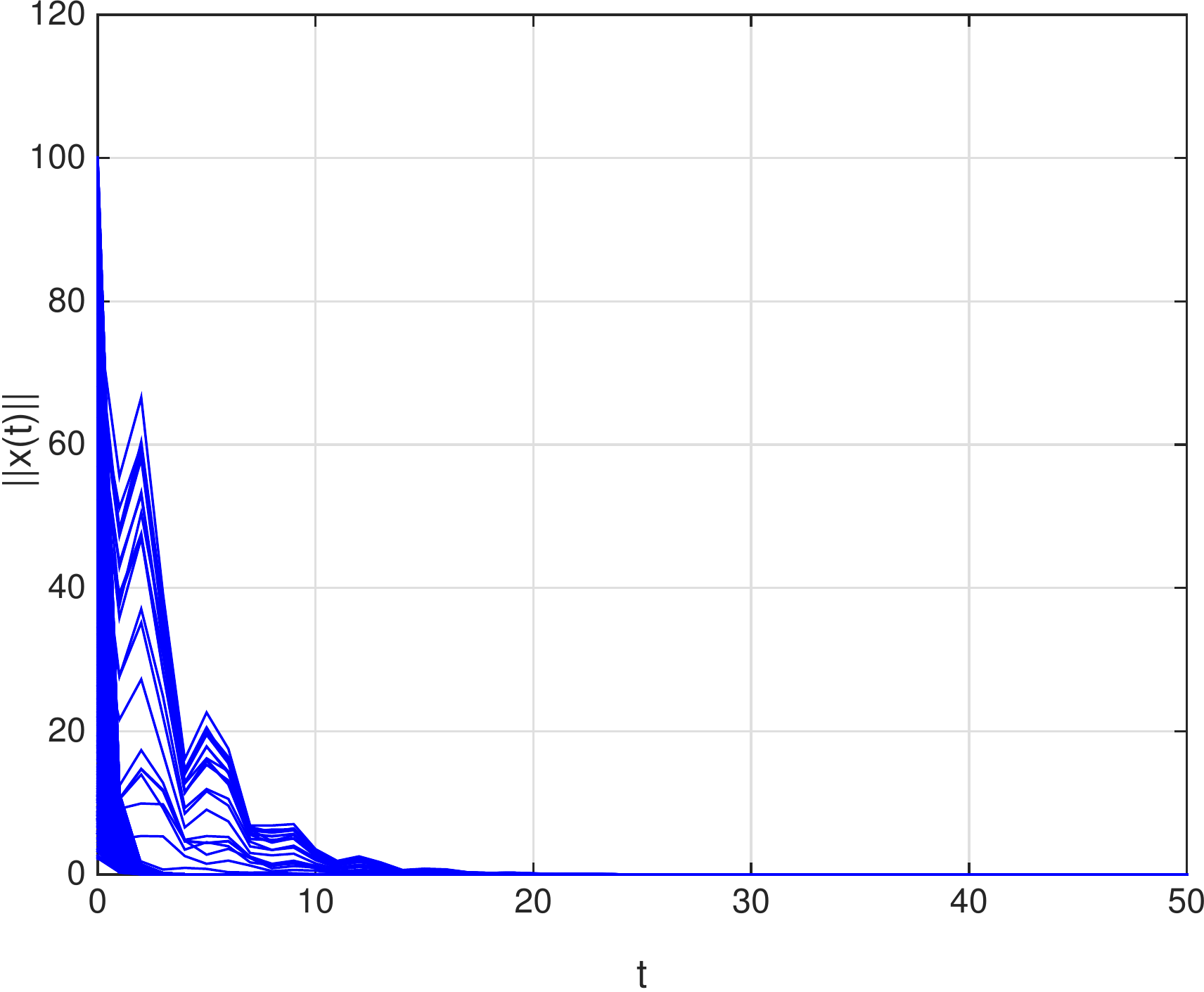}
        \caption{Plot of \((\norm{x(t)})_{t\in\N_{0}}\) with subsystems \(\tilde{A}_{1}\) and \(\tilde{A}_{2}\) under \(\sigma\in\Sd\), \(\Dm = 2\)}\label{fig:ex_xplot_rob}
    \end{minipage}
\end{figure}

\end{document}